\documentclass[12pt]{article}
\usepackage{geometry}
\usepackage{amsmath}
\usepackage{stmaryrd}
\usepackage{amsfonts}
\usepackage{tikz}
\usepackage{appendix}
\usepackage{authblk}
\usepackage{amssymb}
\usepackage{algorithm}
\usepackage{algorithmic}

\usepackage{graphics}
\usepackage{graphicx}
\usepackage{subfigure}

\usepackage{url}
\usepackage{multirow}
\usepackage{eqparbox}
\usepackage{ifthen}
\usepackage{balance}
\usepackage{float}
\usepackage{makecell}
\usepackage{hyperref}
\hypersetup{
	colorlinks   = true, 
	urlcolor     = blue, 
	linkcolor    = blue, 
	citecolor   = blue 
}
\usepackage{amssymb}
\usepackage{xspace}

\newtheorem{definition}{Definition}
\newtheorem{theorem}{Theorem}
\newtheorem{lemma}{Lemma}

\newtheorem{remark}{Remark}

\newcommand{\leastcorevaluegeneral}{${\cal ALCV}(\Gamma)$\xspace}
\newcommand{\rleastcorevaluegeneral}{${\cal RLCV}(\Gamma)$\xspace}
\newcommand{\averagecorevaluegeneral}{${\cal LADV}(\Gamma)$\xspace}

\newcommand{\coreempty}{CORE\xspace}
\newcommand{\alcv}{ALCV\xspace}
\newcommand{\rlcv}{RLCV\xspace}
\newcommand{\ladv}{LADV\xspace}

\newcommand{\gamename}{$\Gamma_{f}$\xspace}
\newcommand{\infgame}{$\Gamma_{\rm inf}$\xspace}
\newcommand{\core}{${\cal C}(\Gamma_{f})$\xspace}

\newcommand{\leastcorevalue}{${\cal ALCV}(\Gamma_{f})$\xspace}

\newcommand{\rleastcorevalue}{${\cal RLCV}(\Gamma_{f})$\xspace}
\newcommand{\averagecorevalue}{${\cal LADV}(\Gamma_{f})$\xspace}

\newcommand{\lpprime}{LP-PRIME\xspace}
\newcommand{\lprelax}{LP-RE\xspace}
\newcommand{\lpstrengthen}{LP-STR\xspace}
\newcommand{\lpofrlcv}{LP-RLCV\xspace}
\newcommand{\lprelaxnew}{LP-RELAX-NEW\xspace}

\newcommand{\inff}{{\sigma^{\rm IC}}}
\newcommand{\OnlyInFull}[1]{#1}
\newcommand{\OnlyInShort}[1]{}

\begin{document}
\title{Coreness of Cooperative Games with Truncated Submodular Profit Functions}

\author[1]{Wei Chen}
\author[2]{Xiaohan Shan}
\author[2]{Xiaoming Sun}
\author[2]{Jialin Zhang}

\affil[1]{Microsoft Research}
\affil[2]{Institute of Computing Technology, Chinese Academy of Sciences}

\maketitle

\begin{abstract}
	{\em Coreness} represents solution concepts related to core in cooperative games, which captures the stability of players.
	Motivated by the scale effect in social networks, economics and other scenario, we study the coreness of cooperative game 
with truncated submodular profit functions.
	Specifically, the profit function 
$f(\cdot)$ is defined 
by a truncation of a submodular function $\sigma(\cdot)$: $f(\cdot)=\sigma(\cdot)$ if $\sigma(\cdot)\geq\eta$ and $f(\cdot)=0$ otherwise, where $\eta$ is a given threshold.
In this paper, we study the core and three core-related concepts of truncated submodular profit cooperative game.
We first prove that whether core is empty can be decided in polynomial time and an allocation in core also can be found in polynomial time when core is not empty.
When core is empty, we show hardness results and approximation algorithms for computing other core-related concepts including 
{\em relative} least-core value, {\em absolute} least-core value and least {\em  average dissatisfaction} value.
\end{abstract}

\section{Introduction}
\OnlyInFull{
With the wide popularity of social media and social network sites such as Facebook, Twitter, WeChat, etc.,
social networks have become a powerful platform for spreading information among individuals.
Thus, influential users always play important role in a social network.
Motivated by this background, influence diffusion in social networks has been extensively studied
 \cite{domingos2001mining,kempe2003maximizing,chen2013information}.}
\OnlyInShort{Influence diffusion in social networks has been extensively studied
	\cite{domingos2001mining,kempe2003maximizing,chen2013information}.}
Most of previous works focus on exploring influential nodes.
To the best of our knowledge, there is no study about the ``stability'' of influential nodes (seed set) when they are treated as a coalition.

Consider the following scenario.
A group of influential people in a social network are considering forming
	a coalition so that they can better serve many advertisers through
	viral marketing in the social network.
To make the coalition stable, we need to design a fair profit allocation
	scheme among the members of the coalition, such that
	no individual or a subset of people have incentive to deviate
	from this coalition, thinking that the allocation to them is unfair and
	they could earn more by the deviation and forming an alliance by
	themselves.
A useful and mature framework of studying such incentives for stable coalition
	formation is the cooperative game theory, and in particular
	the coreness (core and its related concepts) of the cooperative games \cite{Gabrielle2004groupstability,reshef2011ijcai}.

	First we will motivate our consideration of the truncated submodular
	functions here.
In the above social influence scenario, the typical way of measuring the
	contribution of any set $S$ of influential people is by its
	influence spread function $\sigma(S)$, which measures the expected number
	of people in the social network that could be influenced by
	$S$ under some stochastic diffusion model.
Extensive researches have been done on stochastic diffusion models, and
	it has been shown that under a large class of models $\sigma(S)$
	is both monotone and submodular\footnote{A set function
	$f$ is monotone if $f(S) \le f(T)$ for all $S\subseteq T$, and is
	submodular if $f(S\cup \{u\}) - f(S) \ge f(T\cup \{u\}) - f(T)$ for
	all $S\subseteq T$ and $u \not\in T$.
	}
	\cite{kempe2003maximizing,mossel10,chen2013information}.
However, the advertisers would only be interested in the coalition as
	a viral marketing platform when the influence spread reaches certain
	scale level.
In other words, the coalition can only receive profit after the influence
	spread is above a certain scale threshold $\eta$.
Therefore, the true profit function for the coalition is
	$f(S) = \sigma(S)$ when $\sigma(S) \ge \eta$, and $f(S) = 0$ otherwise.
We call such $f$ truncated submodular functions.

Both submodularity and scale effect are common in economic behaviors beyond
	the above example of viral marketing in social networks.
Therefore, considering truncated submodular functions
	as the profit functions is reasonable.
In this paper, we study the computational issues related to the coreness
	of cooperative games with truncated submodular profit functions.

\textbf{Solution Concepts in Cooperative Games.}
A cooperative game $\Gamma=(V, \gamma)$ consists of a player set $V=\{1,2,\cdots,n\}$
and a profit function $\gamma: 2^V\rightarrow \mathbb{R}$
with $\gamma(\emptyset)=0$.
A subset of players $S\subseteq V$ is called a \emph{coalition} and $V$ is called the \emph{grand coalition}.
For each coalition $S$, $\gamma(S)$ represents the profit obtained by $S$ without help of other players.
An allocation over the players is denoted by a vector $x=(x_1,x_2,\cdots,x_n)\in \mathbb{R}^{|V|}$ whose components are one-to-one associated with players in $V$, where $x_i\in \mathbb{R}$ is the value received by player $i\in V$ under allocation $x$.
For any player set $S\subseteq V$, we use the shorthand notation $x(S)=\sum_{i\in S}x_i$. 
A set of all allocations satisfying some specific requirements is called a \emph{solution concept}.

The \emph{core} \cite{gillies1953core,shapley1955core} is one of the earliest and most attractive solution concepts
	that directly addresses the issue of stability.
The core of a game is the set of allocations ensuring that no coalition would have an incentive
    to split from the grand coalition, and do better on its own.
More precisely, the core of a game $\Gamma$ (denoted by ${\cal C}(\Gamma)$), is the following set of allocations:
${\cal C}(\Gamma)$=$\{x\in \mathbb{R}^{|V|}: x(V)=\gamma(V), x(S)\geq\gamma(S),~ \forall~S\subseteq V\}$.
In practice, core is very strict and may be even empty in some cases.
When ${\cal C}(\Gamma)$ is empty, there must be some coalitions becoming dissatisfaction since they can obtain more benefits
    if they leave the grand coalition and work as a separated team.
In this case, we use the dissatisfaction degree
	(or dissatisfaction value), defined as $dv(S,x) = \max\{\gamma(S)-x(S),0\}$,
	to capture the instability of player set $S$ with respect to
	the allocation $x$.
Then, the overall stability of the game can be measured as either the
	worst-case or average-case dissatisfaction degree, for
	which we consider the following three versions.

The first one is the \textit{relative least-core value} (${\cal RLCV}$) \cite{faigle1993relativeleastcore}, which reflects the relative stability,
i.e. the minimum value of the maximum proportional difference between the profits and the payoffs among all coalitions.
\begin{definition}
	Given a cooperative game $\Gamma$, 
	the {\em relative least-core value} of $\Gamma$ (\rleastcorevaluegeneral) is $\min_x\max_S \frac{dv(S,x)}{\gamma(S)}$. Technically, \rleastcorevaluegeneral is the optimal solution of the following linear programming:
	\begin{equation}
	\OnlyInShort{\small}
	\begin{array}{ll}
    \min & \ r\\
	\mbox{s.t.}&\left\{
	\begin{array}{ll}
	x(V)=\gamma(V)\\
	x(S)\geq (1-r)\gamma(S)  & \quad \forall~S\subseteq V\\
	x(\{i\}) \ge 0 & \quad \forall~i\in V
	\end{array}\right.
	\end{array}
	\end{equation}
\end{definition}

The second one is the \textit{absolute least-core value} (${\cal ALCV})$ \cite{Maschler1979leastcore} which reflects the absolute stability,
    i.e. the minimum value of the maximum difference between the profits and the payoffs among all coalitions.
The formal definition is as following.
\begin{definition}
Given an cooperative game $\Gamma$,
the {\em absolute least-core value} of $\Gamma$ (\leastcorevaluegeneral) is $\min_x\max_Sdv(S,x)$. Technically, \leastcorevaluegeneral is the optimal solution of the following linear programming:
 \begin{equation}
 \OnlyInShort{\small}
 \begin{array}{ll}
\min & \ \varepsilon\\
 \mbox{s.t.}&\left\{
 \begin{array}{ll}
 x(V)=\gamma(V)\\
 x(S)\geq\gamma(S)-\varepsilon  & \quad \forall~S\subseteq V\\
 	x(\{i\}) \ge 0 & \quad \forall~i\in V
 \end{array}\right.
 \end{array}
 \end{equation}
\end{definition}

The above two classical least-core values capture the stability from the perspective of the most dissatisfied coalition i.e. the worst case of stability.
Sometimes the worst case is too extreme to reflect the real stability.
Thus, we introduce the \textit{least average dissatisfaction value} (${\cal LADV}$) which reflects the minimum value of average dissatisfaction degree among all coalitions.

\begin{definition}\label{def:acv}
	Given a cooperative game $\Gamma$, the
	{\em least average dissatisfaction value} of $\Gamma$ (\averagecorevaluegeneral) is
$\min_x \mathbb{E}_S(dv(S,x))$.
 Technically, \averagecorevaluegeneral is the optimal value of
	the following linear programming:
	\begin{equation}
	\OnlyInShort{\small}
	\begin{array}{ll}
	\min & \frac{1}{2^n}\sum_{S\subseteq V}\max \{\gamma(S)-x(S), 0\}\\
	\mbox{s.t.}&\left\{
	\begin{array}{ll}
	x(V)=\gamma(V)\\
	x(\{i\}) \ge 0 & \quad \forall~i\in V\\
	\end{array}\right.
	\end{array}
	\end{equation}
\end{definition}

In this paper, we consider the following computational problems in the context
	of truncated submodular functions:
(a) Whether the core of a given cooperative game is empty?
(b) How to find an allocation in core if the core is not empty?
(c) If the core is empty, how
	to compute the relative least-core value, the absolute least-core value and the least average dissatisfaction value of a cooperative game?

\textbf{Contributions.}
We study coreness (solution concepts related to core) of truncated submodular profit cooperative game \gamename.
We consider computational properties of the core, the relative least-core value, the absolute least-core value and the least average dissatisfaction value of \gamename, which are denoted by \core, \rleastcorevalue, \leastcorevalue and \averagecorevalue, respectively.

We first prove that checking the non-emptiness of \core can be done in polynomial time.
Moreover, we can find an allocation in the core if the core is not empty.
Next, we consider the case when the core is empty.
For the problem of computing the relative least-core value (\rleastcorevalue),
	we show that it is in general NP-hard, but
	when truncation threshold $\eta=0$, there is a polynomial time
	algorithm.
Along the way, we also find an interesting partial result showing that
	there is no polynomial time separation oracle for the \rleastcorevalue's
	linear program unless P=NP, which is of independent interest since
	it reveals close connections with a new class of combinatorial problems.
%
%
For the absolute least-core value problem \leastcorevalue, we prove that
	finding \leastcorevalue is APX-hard even when $\sigma(\cdot)$ is defined as the influence spread under the classical independent cascade
	(IC) model in social network.
We also prove that there exists a polynomial time algorithm which can guarantee an additive term approximation.
Finally, for the least average dissatisfaction value problem \averagecorevalue,
	we show that we can use the stochastic gradient descent algorithm to compute \averagecorevalue to an arbitrary small additive error.
	
\textbf{Related Work.}
Cooperative game theory is a branch of (micro-)economics that studies the behavior of self-interested agents in strategic settings where binding agreements between agents are possible \cite{chalkiadakis2011computational}.
Numerous classical studies about cooperative game provide rich mathematical framework to solve issues related to cooperation in multi-agent systems \cite{deng1994complexity,ieong2005marginal,conitzer2006complexity}.
\cite{schulz2013approximating} studies the approximation of the absolute least core value of supermodular cost cooperative games, the results in this paper can be generalized to submodular profit cooperative games.
An important application of our study is to analyze the stability of influential people in social networks.
Almost all the existing studies focus on selecting seed set \cite{ChenWY09efficientinfluence,goyal2012minimizing,tang2015rrset}.
To the best of our knowledge, there is no literature considering the stability of the selected seed set.
We utilize cooperative game theory to analyse the stability of seed set,
	and generalize it to a generic cooperative game with truncated submodular functions. 
	The truncated operation represents the  ``threshold effect'' which has been studied widely in literature\cite{Mark1978threshold,Albert1984}. 
	
\section{Model and Problems}
\subsection{Cooperative Games with Truncated Submodular Profit Functions}
A truncated submodular profit cooperative game is denoted by \gamename$=(V, f(\cdot))$.
\OnlyInShort{In \gamename, $V$ is the player set and $f(\cdot)$ is the profit function which is defined as $f(S)=\sigma(S)$ if $\sigma(S)\geq \eta$ and $f(S)=0$ if $\sigma(S)< \eta$.}
In \gamename, $V$ is the player set and $f(\cdot)$ is the profit function which is defined as follows:
\begin{equation*}
\OnlyInShort{\small}
f(S)=
\begin{cases}
\sigma(S), & \text{if $\sigma(S)\geq \eta$}\\
0, & \text{if $\sigma(S)< \eta$}
\end{cases}
\end{equation*}
Note that $\sigma(\cdot)$ is a nonnegative monotone increasing submodular function with $\sigma(\emptyset)=0$ and $0\leq\eta\leq\sigma(V)$ is a nonnegative threshold.
To express clearly, in the left of this paper, a truncated submodular profit cooperative game is denoted by a triple form $(V, \sigma(\cdot), \eta)$.

Note that the explicit representation of $\sigma(\cdot)$ might be exponential in the size of $V$.
The standard way to bypass this difficulty is to assume that
 $\sigma(\cdot)$ is given by a value oracle.

\subsection{Computational Problems on the Coreness}\label{sec:problemdefinition}
Given an truncated submodular profit cooperative game \gamename, we focus on the following problems:\\
\coreempty: Is \core $\neq\emptyset$ and how to find an allocation in \core when \core$\neq\emptyset$?\\
\alcv: When \core$=\emptyset$, how to compute \leastcorevalue?\\
\rlcv: When \core$=\emptyset$, how to compute \rleastcorevalue?\\
\ladv: When \core$=\emptyset$, how to compute \averagecorevalue?

Before we analyze the above problems, we introduce a specific instance of truncated submodular profit cooperative game (see Section \ref{sec:infgame}).

\subsection{Influence Cooperative Game (\infgame)}\label{sec:infgame}
As the description in our introduction, an important motivation of our model is influence in social networks.
In this section, we introduce a specific instance of truncated submodular profit cooperative game, {\em influence cooperative game}.

{\bf Social graph.} A social graph is a directed graph $G=(V\cup U,E; P)$, where $V\cup U$ is the vertex set and $E$ is the edge set.
$P=\{p_e\}_{e\in E}$ and $p_e$ is the influence probability on each edge $e\in E$.
Note that, $V$ and $U$ denote the vertex set of influential people and target people in $G$, respectively.

{\bf Influence diffusion model.} The information diffusion process follows the independent cascade (IC) model proposed by \cite{kempe2003maximizing}.
\OnlyInFull{
In the IC model, discrete time steps $t=0, 1, 2, \cdots$ are used to model the diffusion process. Each node in $G$ has two states:
inactive or active.
 At step 0,
nodes in seed set $S$ are active and other nodes are inactive.
For any step $t\ge 1$, if a node $u$ is newly active at step $t-1$,
 $u$ has a single chance to influence each of its inactive out-neighbor $v$ with independent probability $p_{uv}$
to make $v$ active.
Once a node becomes active, it will never return to the inactive state.
The diffusion process stops when there is no new active nodes at a time step.
}
For any $S\subseteq V$, we use $\inff(S)$ to denote the influence spread of $S$,
	the expected number of activated nodes in $U$ from seed set $S\subseteq V$, at the end of an IC diffusion.
According to \cite{kempe2003maximizing}, $\inff(\cdot)$ is a monotone submodular function.

\begin{definition}
An {\em influence cooperative game} \infgame$=(V, \inff(\cdot), \eta)$ is a special form of the truncated cooperative
	game, with $V$ as the player set,
	and the truncation of influence spread function $\inff(\cdot)$ as the profit function.
\end{definition}

In the rest of this paper, we analyze problems defined in Section \ref{sec:problemdefinition} one by one.
Note that our positive results (properties and algorithms) could apply to all truncated submodular profit cooperative games including influence cooperative game.
Our hardness results are established for the influence
	cooperative games, so it is stronger than the hardness results
	for general truncated submodular cooperative games.

\section{Computing Core}\label{sec:core}
We start by considering the core of \gamename (\core).
In \gamename, we say a player $i\in V$ is a \emph{veto player}
    if $\sigma(S)<\eta$ for any $S\subseteq V\setminus\{i\}$.
That is to say, a successful coalition must include all veto players.
\begin{lemma}\label{lem:coreiff}
\core $\neq \emptyset$ if and only if:\\
   (i) There exists at least one veto player in \gamename, or\\
   (ii) $\sigma(S)=\sum_{i\in S}\sigma(\{i\})$, for any $S\subseteq V$.\\
\end{lemma}

\begin{proof}
Suppose the player set of \gamename is $V=\{1,2,\cdots, n\}$. We first prove the sufficiency of Lemma \ref{lem:coreiff}.
On one hand, suppose $i$ is a veto player of \gamename, then we can find a trivial allocation $x$ in \core:
    $x(\{i\})=\sigma(V)$ and $x(\{j\})=0$, $\forall~j\in V\setminus \{i\}$.
On the other hand, $x(\{i\})=\sigma(\{i\})$ ( $\forall i \in V$) is an allocation in \core if $\sigma(S)=\sum_{i\in S}\sigma(\{i\})$.

Now we prove the necessity.
Suppose \core$\neq\emptyset$ and $x\in$\core.
Let $\sigma(V)=\sum_{i=1}^n M_i$, where $M_i=\sigma(\{1,2,\cdots,i\})-\sigma(\{1,2,\cdots,i-1\})$
    is the marginal increasing of player $i$.
If there is no veto player, then for any $i\in V$, $\sigma(V\setminus\{i\})\geq\eta$
    since $\sigma(S)$ is monotone.
Thus, $f(V\setminus\{i\})=\sigma(V\setminus\{i\})$, $\forall~i\in V$.
Suppose $\sigma(V\setminus\{i\})=\sum_{j=1}^{i-1}M_j+\sum_{j=i+1}^{n}M'_{ij}$,
    where $M'_{ij}=\sigma(\{1,2,\cdots,i-1,i+1,\cdots,j\})-\sigma(\{1,2,\cdots,i-1,i+1,\cdots,j-1\})$.
Note that $M'_{ij}\geq M_j$ since $\sigma(S)$ is submodular.
By the definition of the core, for any $i\in\{1,2,\cdots,n\}$, we have:
$x(V\setminus\{i\})\geq f(V\setminus\{i\})=\sigma(\{V\setminus\{i\}\})$.
That is,$x(V)-x(\{i\})\geq \sum_{j=1}^{i-1}M_j+\sum_{j=i+1}^{n}M'_{ij}$, $\forall i \in V$.

Summing up these inequalities for all $i\in V$, we have,
$(n-1)\sum_{i=1}^n x(\{i\})\geq\sum_{i=1 }^n (\sum_{j=1}^{i-1}M_j+\sum_{j=i+1}^{n}M'_{ij})
	\geq \sum_{i=1 }^n(\sum_{j=1}^{i-1}M_j+\sum_{j=i+1}^{n}M_j)
	=\sum_{i=1 }^n(\sigma(V)-M_i)
	=(n-1)\sigma(V)$.


We have known that $\sum_{i=1}^n x(\{i\})=\sum_{j=1}^{n}M_j=\sigma(V)$ and then
    $M_j=M'_{ij}$, $\forall i,j\in V$.
Thus, $\sigma(S)=\sum_{i\in S}\sigma(\{i\})$.
\end{proof}
An important application of Lemma \ref{lem:coreiff} is Theorem \ref{thm:core_polysolvable}.
\begin{theorem}\label{thm:core_polysolvable}
Deciding whether \core is empty can be done in polynomial time and
 an allocation in \core can be computed in polynomial time if \core is not empty.
\end{theorem}
\begin{proof}[Sketch]
	First, it takes polynomial time to check the non-emptiness of \core. When \core is not empty, then $(x_j=\sigma(V), \textbf{0}_{\{i:i\neq j\}})$ $\in$ \core when $j$ is a veto player and $(\sigma(\{1\}), \cdots, \sigma(\{n\}))$ $\in$ \core when $(ii)$ satisfies.
\end{proof}
The detail proof of Theorem \ref{thm:core_polysolvable} is shown in the appendix.

\section{Computing Relative Least-Core Value}
\label{sec:relativeleastcorevalue}
From Lemma \ref{lem:coreiff}, \core may be empty in many cases.
It is obvious that \rleastcorevalue$>0$ if \core$=\emptyset$ and \rleastcorevalue$=0$ otherwise.
In this section, we study computational properties of \rlcv  problem.
The linear programming corresponding to \rleastcorevalue (\lpofrlcv) is as follows:
\begin{equation}
\OnlyInShort{\small}
\begin{array}{ll}
\min& r\\
\mbox{s.t.}&\left\{
\begin{array}{ll}
x(V)=\sigma(V)\\
x(S)\ge (1-r)\sigma(S)  & \forall~S\subseteq V,~ \sigma(S)\geq \eta\\
x(\{i\}) \ge 0 & \forall~i\in V\\
\end{array}\right.
\end{array}
\end{equation}

A special case of computing \rleastcorevalue is when $\eta=0$.
It captures the scenario that the profit of any coalition exactly equals to its influence spread under influence cooperative game.
In Theorem~\ref{thm:existseparationoracle} we show that, although there are exponential number of constraints, \lpofrlcv can be solved in polynomial time by
providing a polynomial time separation oracle when $\eta=0$. 
A separation oracle for a linear program is an algorithm that, given a putative feasible solution, checks whether it is indeed feasible, and if not, outputs a violated constraint.
It is known that a linear program can be solved in polynomial time by the ellipsoid method as long as it has a polynomial time separation oracle \cite{grotschel2012geometric}.
\begin{theorem}\label{thm:existseparationoracle}
	There exists a polynomial time separation oracle of \lpofrlcv when $\eta=0$.
Therefore, \rlcv can be solved in polynomial time when $\eta=0$.
\end{theorem}

\begin{proof}
	Given any solution candidate of \lpofrlcv $(x',r')$, we need to either assert $(x',r')$ is a feasible solution or
	find a constraint in \lpofrlcv such that $(x',r')$ violates it.
	Note that, checking $x'(V)=\sigma(V)$ and $x'(\{i\})\ge 0$ ($\forall~ i\in V$) can be done
	in polynomial time.
	Thus, we only need to check whether $g(S)\triangleq 1-x'(S)/\sigma(S)\leq r',~ \forall S\subseteq V$.
	
	An important property is 
	 $g(S)$ achieves its maximum value when $S$ contains only one single player.
This is because
 $g(S)=1-\frac{x'(S)}{\sigma(S)}\leq 1-\frac{\sum_{i\in S}x'_i}{\sum_{i\in S}\sigma(\{i\})}\leq 1-\min_{i:i\in S}\{\frac{x'_i}{\sigma(\{i\})}\}=\max_{i:i\in S}\{g(\{i\})\}$.
	The first inequality is due to the submodularity of $\sigma(S)$ and the second inequality is due to
	$\min_{i:i\in[n]}\{\frac{a_i}{b_i}\}\leq \frac{\sum_{i=1}^n a_i}{\sum_{i=1}^n b_i}$, $\forall a_i, b_i \in \mathbb{R}$.
	Thus, the exponential number of constraints can be simplified to $n$ constraints on all single players.
	Then, we can find a polynomial time separation oracle of \lpofrlcv directly.
\end{proof}

When $\eta=0$, \rlcv can be solved in polynomial time is mainly because the most dissatisfaction coalition is a single player.
However, when $\eta\neq 0$, it becomes intractable to find the most dissatisfaction coalition.

\begin{theorem}\label{thm:oraclehard}
	There is no polynomial time separation oracle of \lpofrlcv for some $\eta>0$, unless P=NP.
\end{theorem}

Theorem \ref{thm:oraclehard} can not imply the NP-hardness of \rlcv.
However, the proof of Theorem \ref{thm:oraclehard} reveals an interesting connection between \rlcv problem and a series of well defined combinatorial problems.
We will report the proof of Theorem \ref{thm:oraclehard} and the generalized combinatorial problems in the appendix.

In the left of this section, we prove the NP-hardness of \rlcv, a stronger hardness result than which in Theorem \ref{thm:oraclehard}.
\begin{theorem}\label{thm:rleastcorenphard}
	It is NP-hard to compute \rleastcorevalue, even under influence cooperative game.
\end{theorem}

\begin{proof}[Sketch]
	We construct a reduction from the SAT problem.
	A boolean formula is in conjunctive normal form (CNF) if it is expressed as an AND of clauses,
	each of which is the OR of one or more literals.
	The SAT problem is defined as follows:
	given a CNF formula $F$, determine whether $F$ has a satisfiable assignment.
	Let $F$ be a CNF formula with $m$ clauses $C_1, C_2, \cdots, C_m$, over $n$ literals $z_1, z_2, \cdots, z_n$.
	Without loss of generality, we set $m>4n$.
	
	We construct a social graph $G$ as follows:
	\begin{figure}[t]
		\centering
		\includegraphics[width=0.6\textwidth]{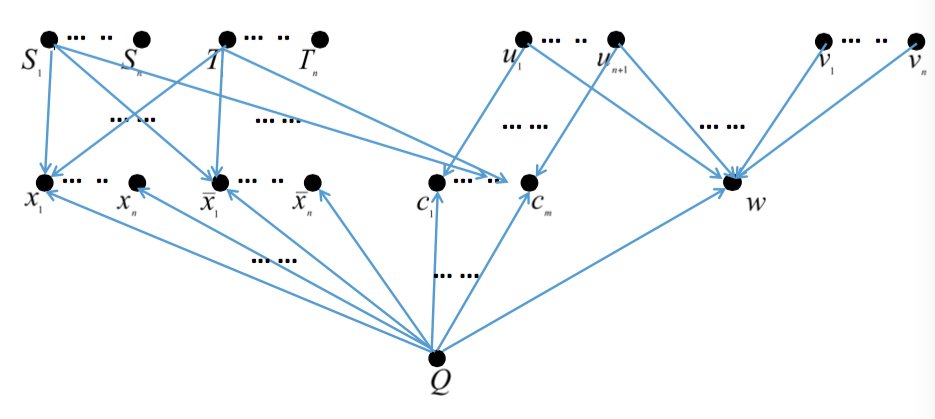}\\
		\caption{The reduction from SAT to \rleastcorevalue}\label{fig:reduction-sat}
	\end{figure}
	$G=(V_1\cup V_2\cup V_3, E)$ is a tripartite graph (see the sketch graph in Figure \ref{fig:reduction-sat}).
	In the first layer ($V_1$), there are two nodes $S_i$ and $T_i$ corresponding to each $i\in \{1, 2, \cdots, n\}$,
	$n+1$ dummy nodes labeled as $u_1, u_2, \cdots, u_{n+1}$ and $n$ dummy nodes labeled as $v_1, v_2, \cdots, v_n$.
	In the second layer ($V_2$), there are two nodes $x_i$ and $\overline{x}_i$ corresponding to each $i\in \{1, 2, \cdots, n\}$, one node $c_j$ for each $j\in\{1, 2, \cdots, m\}$ and a dummy node $w$.
	The third layer ($V_3$) contains only node $Q$.
	Edges exist only between the adjacent layers.
	For each $i\in \{1, 2, \cdots, n\}$, $S_i$ sends an edge to every node in
	$\{x_i, \overline{x}_i\}\cup\{c_j:$ clause $C_j$ contains literal $z_i, j\in\{1, 2, \cdots, m\}\}$.
	Similarly, for each $i\in \{1, 2, \cdots, n\}$, $T_i$ sends an edge to every node in
	$\{x_i, \overline{x}_i\}\cup\{c_j:$ clause $C_j$ contains literal $\overline{z}_i, j\in\{1, 2, \cdots, m\}\}$.
	The probabilities on edges sent form $S_i$ and $T_i$ are 1.
	There is an edge with influence probability 1 from $u_i$ to $c_i$ for any $i\in \{1, 2, \cdots, n\}$ and $m-n$ edges form $u_{n+1}$ to $c_{n+1}, c_{n+2}, \cdots, c_m$.
	There is an edge from $u_i$ to $w$ with influence probability $1-\sqrt[n+1]{1/2}$ for any $i\in \{1, 2, \cdots, n+1\}$.
	There is also exists an edge from $v_i$ to $w$ with influence probability $1-\sqrt[n]{1/2}$ for any $i\in \{1, 2, \cdots, n\}$.
	The left edges are  from $Q$ to all nodes in the second layer.
	The influence probability on edge $(Q, w)$ is $1/2$ and all other probabilities on edges sent from $Q$ is 1.
	The influence cooperative game defined on $G$ is $\Gamma(G)=(V=V_1\cup V_3, \inff(\cdot), \eta=2n+m+1/2)$.
	For convenient, we set $N=2n+m$.
	
	Suppose $r^*$ is the optimal solution of the relative least-core value of  $\Gamma(G)$
	We can prove that $r^*\geq 1-\frac{1}{3}(N+\frac{7}{8})/(N+\frac{1}{2})$  if $F$ is satisfiable and $r<1-\frac{1}{3}(N+\frac{7}{8})/(N+\frac{1}{2})$ if $F$ is un-satisfiable.
	The proof of this part is shown in the appendix.
\end{proof}

\section{Computing Absolute Least-Core Value}\label{sec:leastcorevalue}

\subsection{Hardness of \alcv}\label{sec:hardnessofleastcorevalue}
\begin{theorem}\label{thm:alcvhardness}
\alcv problem of influence cooperative game cannot be approximated within 1.139 under the unique games conjecture.
\end{theorem}

\begin{proof}[Sketch]
	We construct a reduction from MAX-CUT problem. Under our construction, for any instance of MAX-CUT problem, we can construct an instance of  \alcv problem such that the optimal solution of these two instances are equal. The detail proof is shown in our appendix. 
\end{proof}

\subsection{Approximating \leastcorevalue}\label{sec:approximateleastcorevalue}

In this section, we approximate \leastcorevalue by approximating the following linear programming (\lpprime):

\begin{equation*}\label{equ:lpprime}
\OnlyInShort{\small}
\begin{array}{ll}
\mbox{min}& \ \varepsilon\\
\mbox{s.t.}&\left\{
\begin{array}{ll}
 x(V)=\sigma(V)\\
 x(\{S\})\ge \sigma(\{S\})-\varepsilon  & \quad \forall S\subseteq V, \sigma(S)\geq \eta\\
 x(\{u\}) \ge 0 & \quad \forall u\in V\\
\end{array}\right.
\end{array}
\end{equation*}

The intractability of \lpprime lies on the exponential number of constraints and
    the hardness of identifying all successful coalitions.
We use a relaxed version \lprelax and a strengthen version \lpstrengthen of \lpprime to
    design an approximation algorithm of \leastcorevalue.
(\ref{equ:lprelax}) and (\ref{equ:lpstrengthen}) are formal definitions of \lprelax and \lpstrengthen, respectively.

\begin{equation}\label{equ:lprelax}
\OnlyInShort{\small}
\begin{array}{ll}
\mbox{min}& \ \varepsilon\\
\mbox{s.t.}&\left\{
\begin{array}{ll}
 x(V)=\sigma(V)\\
 x(S)\ge \eta-\varepsilon  & \quad \forall~S\subseteq V, \sigma(S)\geq \eta\\
 x(\{u\}) \ge 0 & \quad \forall~u\in V\\
\end{array}\right.
\end{array}
\end{equation}

\begin{equation}\label{equ:lpstrengthen}
\OnlyInShort{\small}
\begin{array}{ll}
\mbox{min}& \ \varepsilon\\
\mbox{s.t.}&\left\{
\begin{array}{ll}
 x(V)=\sigma(V)\\
 x(S)\ge \sigma(S)-\varepsilon  & \quad \forall~S\subseteq V\\
 x(\{u\}) \ge 0 & \quad \forall~u\in V\\
\end{array}\right.
\end{array}
\end{equation}

	Intuitively, \lprelax and \lpstrengthen denote absolute least-core values of two cooperative games with new profit functions.
	Specifically, \lprelax relaxes the constraints in \lpprime
	by reducing the profits of all successful coalitions excepting $V$ to $\eta$.
Formally, the profit function in \lprelax is $g(S)$:
$g(V)=\sigma(V)$, $\forall~S\subset V$,  $g(S)=\eta$ if $\sigma(S)\geq\eta$ and $g(S)=0$ otherwise.
The profit function in \lpstrengthen is $h(S)=\sigma(S)$, $\forall S\subseteq V$.
Clearly, \lpstrengthen strengthens \lpprime by increasing the profits of all unsuccessful coalitions.

Our main result in this section is shown in Theorem \ref{thm:approxmate_alcv}.

\begin{theorem}\label{thm:approxmate_alcv}
	$\forall~\delta>0$, there exists an approximate algorithm ${\cal A}$ of the \leastcorevalue problem  with running time in $poly(n, 1/\delta, \log\sigma(V))$,
	${\cal A}$ outputs $\varepsilon'_p$ such that $\varepsilon^*_p\leq \varepsilon_p'\leq\min\{\varepsilon^*_p+\sigma(V)-\eta+2\delta,~\max\{3\varepsilon^*_p,\eta\}\}$,.
	
\end{theorem}

We prove Theorem \ref{thm:approxmate_alcv} by show Lemma \ref{lem:lpplprlps}, Lemma \ref{lem:lpstrengthen} and Lemma \ref{lem:lprelax} in order.
\begin{lemma}\label{lem:lpplprlps}
	Suppose the optimal value of \lpprime, \lprelax and \lpstrengthen are $\varepsilon_p^*$, $\varepsilon_r^*$ and $\varepsilon_s^*$, respectively.
	Then, we have
	\begin{equation}\label{eq:lprandlpp}
	\varepsilon_p^*\leq \varepsilon_r^*+(\sigma(V)-\eta)\leq \varepsilon_p^*+(\sigma(V)-\eta),
	\end{equation}
	\begin{equation}\label{eq:lpsandlpp}
	\varepsilon_p^*\leq \varepsilon_s^*\leq \max\{\varepsilon_p^*,\eta\}.
\end{equation}
\end{lemma}

\begin{lemma}\label{lem:lpstrengthen}
	There exists a polynomial time approximate algorithm of \lpstrengthen outputting $\varepsilon_s'$ such that $\varepsilon_s^*\leq\varepsilon_s'\leq 3\varepsilon_s^*$.
\end{lemma}

\begin{lemma}\label{lem:lprelax}
	$\forall~\delta>0$, there exists an algorithm of \lprelax outputting $\varepsilon_r'$
	such that $\varepsilon_r^*\leq\varepsilon_r'\leq\varepsilon^*_r+2\delta$, with runs time in $poly(n, 1/\delta, \log \sigma(V))$.
\end{lemma}
The proofs of Lemma \ref{lem:lpplprlps} to Lemma \ref{lem:lprelax} rely heavily on mathematical computation and we report them in the appendix.

\section{Computing Least Average Dissatisfaction Value}\label{sec:ladv}
Based on Definition \ref{def:acv}, \averagecorevalue equals the optimal value of
    the following linear programming:
\begin{equation}\label{equ:ladv}
\OnlyInShort{\small}
\begin{array}{ll}
\mbox{min}& F(x)=\frac{1}{2^n}\sum_{S\subseteq V}\max \{f(S)-x(S), 0\}\\
\mbox{s.t.}&\left\{
\begin{array}{ll}
 x(V)=\sigma(V)\\
 x(\{i\}) \ge 0 & \quad \forall~i\in V\\
\end{array}\right.
\end{array}
\end{equation}
Where $f(S)=\sigma(S)$ if $\sigma(S)\geq \eta$ and $f(S)=0$ otherwise.
There are exponential terms in $F(x)$, however, we can utilize stochastic gradient algorithm to approximate the optimal solution of (\ref{equ:ladv}).
This is because the object function $F(x)$ is a convex function (Lemma \ref{lem:gsubmodular}) and the feasible solution area in (\ref{equ:ladv}) is a convex set.

\begin{lemma}\label{lem:gsubmodular}
	$F(x)$ is a convex function.
\end{lemma}
The proof of Lemma \ref{lem:gsubmodular} is shown in our appendix.
The stochastic gradient descent algorithm
	(cf. \cite{shalev2014understanding}) can be used to compute \averagecorevalue (see Algorithm \ref{alg:sgd}).

\begin{algorithm}[h]
	\OnlyInShort{\small}
\renewcommand{\algorithmicrequire}{\textbf{Input:}}
\renewcommand\algorithmicensure {\textbf{Output:}}
\caption{Stochastic gradient descent for \ladv}\label{alg:acv}
\begin{algorithmic}[1]\label{alg:sgd}
\STATE	\textbf{Parameters:} Scaler $\alpha>0$, integer $T>0$
\STATE  \textbf{Initialize:} $\textbf{X}^1=\textbf{0}$, $t=0$.\\
\STATE  Set $D=\{\textbf{X}:\textbf{X}_i\geq 0 (\forall~i\in V),~\sum_{i\in V}\textbf{X}_i= \sigma(V)\}$.\\
\FOR {$t=1$ to $T$}
	\STATE  /*choose a random $\textbf{Y}^t$ such that $\mathbb{E}[\textbf{Y}^t|\textbf{X}^t]$ is a subgradient of $F$.*/
	\STATE	Uniformly at random choose a set $S\in 2^V$.\\ \label{line:beginchoose}
	\IF {$f(S)\geq\textbf{X}^t(S)$}
		\STATE Set $\textbf{Y}^t=(-\textbf{1}_S, \textbf{0}_{V\setminus S})$.\\
	\ELSE
			\STATE Set $\textbf{Y}^t=\textbf{0}$.
	\ENDIF \label{line:endchoose}
    \STATE  update $\textbf{X}^{t+\frac{1}{2}}=\textbf{X}^t-\alpha \textbf{Y}^t$.\\
    \STATE  /*Project $\textbf{X}^{t+\frac{1}{2}}$ to D*/
    \STATE  $\textbf{X}^{t+1}=\arg\min_{\textbf{X}\in D} \|\textbf{X}-\textbf{X}^{t+\frac{1}{2}}\|^2 $.\\
\ENDFOR
\RETURN $\hat{F}=\min \{F(\textbf{X}^t)\}_{t\in \{1,2,\cdots, T\}}$.
\end{algorithmic}
\end{algorithm}


Let $F^*$ be the optimal solution of \averagecorevalue, $\hat {F}$ be the output of Algorithm \ref{alg:sgd} and the profit of grand coalition $\sigma(V)=V$. Then, the performance of Algorithm \ref{alg:sgd} can be formalized in the following theorem.
\begin{theorem}\label{thm:sgd}
	$\forall$ $\varepsilon>0$, $\mathbb{E}[\hat {F}]-F^*\leq \varepsilon$ if $T\geq\frac{\sigma(V)^4n^4}{\varepsilon^2}$ and $\alpha=\sqrt{\frac{\sigma(V)^4}{Tn^4}}$ in Algorithm \ref{alg:acv}.
\end{theorem}
Following  the analysis in Chapter 14 of \cite{shalev2014understanding},
Theorem \ref{thm:sgd} holds since it is easy to check that $\mathbb{E}[\textbf{Y}^t|\textbf{X}^t]$ is a subgradient of $F(\textbf{X})$ at node $\textbf{X}^t$, for any $t\in [T]$ (line \ref{line:beginchoose} - line\ref{line:endchoose} in Algorithm \ref{alg:sgd}).
\section{Conclusion and future work}
In this paper, we study the core related solution concepts of truncated submodular profit cooperative game.
One possible future work is to change the way of truncating a function.
For example, we can set $f(S)=\sigma(S)$ if $|S|\geq k$ and $f(S)=0$ otherwise.
This setting is a special case of the setting in our paper and maybe we can try to design algorithms for it.
In this paper, we prove that computing the relative least-core value is NP-hard.
We also prove that the relative least-core value can be solved in polynomial time in a special case.
A directly future work is to design an approximate algorithm of \rlcv under general case.

\newpage
\appendix
\section{Appendix of section \ref{sec:core}}

\begin{proof}[Proof of Theorem \ref{thm:core_polysolvable}]
	We can design the following polynomial time process to check the emptiness of \core and find an allocation in \core when \core$\neq\emptyset$.\\
	\textbf{Step 1:} Query $\sigma(V\setminus\{i\})$ for all $i\in V$ from value oracle. If there exists $j\in V$ such that $\sigma(V\setminus\{j\})<\eta$, go to Step 2, otherwise, go to Step 3.\\
	\textbf{Step 2:} Return $x=(0, \cdots,x_j=\sigma(V), 0, \cdots, 0)$ $\in$ \core.\\
	\textbf{Step 3:} Query $\sigma(V)$ and $\sigma(\{i\})$ for all $i\in V$ from value oracle. If $\sigma(V)=\sum_{i\in V}\sigma(\{i\})$, go to Step 4, otherwise, go to Step 5.\\
	\textbf{Step 4:} Return $x=(\sigma(\{1\}), \cdots, \sigma(\{n\}))$ $\in$ \core.\\
	\textbf{Step 5:} Assert that \core$=\emptyset$.
\end{proof}

\section{Appendix of Section \ref{sec:relativeleastcorevalue}}
\begin{proof}[Proof of Theorem \ref{thm:oraclehard}]
		We construct a reduction from NP-complete problem DOMINANT-SET~\cite{michael1979computers}. Given an undirected graph $G=(V,E)$ and an integer $k\in \mathbb{N}$, the DOMINANT-SET problem concerns testing whether there exists a dominant set of $G$ with size no more than $k$. A dominant set is a subset $S\subseteq V$ such that each vertex in $V\setminus S$ is adjacent to at least one vertex in $S$.

	Given any instance of DOMINANT-SET problem $(G=(V,E); k)$, we construct a social graph $G'$ as follows:
	The vertex set in $G'$ is $V'=V_1\cup V_2$, where $V_1=V_2=V$.
	For each node $i\in V_1$ and $j\in V_2$, there is a directed edge $(i,j)$ in $G'$ if and only if
	either $(i,j)\in E$ in $G$ or $i=j$.
	The influence probability on each edge is 1.
	
	The influence cooperative game defined on $G'$ is $\Gamma(G')=(V_1,\inff(\cdot ), \eta=|V_1|=n)$.
	Thus, the linear programming corresponding to the relative least-core value of $\Gamma(G')$ is:
	\begin{equation}\label{equ:oraclenphard}
	\OnlyInShort{\small}
	\begin{array}{ll}
	\mbox{min}& r\\
	\mbox{s.t.}&\left\{
	\begin{array}{ll}
	x(V_1)=n\\
	x(S)\ge (1-r)n  & \forall~S\subseteq V_1,~ \inff(S)=n\\
	x(\{i\}) \ge 0 &  \forall~i\in V_1\\
	\end{array}\right.
	\end{array}
	\end{equation}
	
	Now we prove that the DOMINANT-SET problem can be solved in polynomial time
	if there exists a polynomial time separation oracle of (\ref{equ:oraclenphard}).
	Given a candidate solution $(x',r')$, where $x'_i=1$ for any $i\in V_1$ and $r'=1-(k+1)/n$.
	Suppose there exists a polynomial time separation oracle $\mathcal{O}$ of (\ref{equ:oraclenphard}).
	Then $\forall~S\subseteq V_1,~ \inff(S)=n$, we can decide whether $|S|\ge (\frac{k+1}{n})n=k+1$ in polynomial time.
	Note that $\{S: S\subseteq V_1,~ \inff(S)=n\}$ is the set of all dominant sets of $G$.
	Thus, for any $G$'s dominant set $S$, $(x', r')$ is a feasible solution if and only if $|S|\ge k+1$.
	In other words, having $\mathcal {O}$, we can decide whether there exists a dominant set with size no more than $k$.
\end{proof}

In remark \ref{rem:adversarial}, we introduce a class of combinatorial optimization problems inspired from the proof process of Theorem \ref{thm:oraclehard}.

\begin{remark}\label{rem:adversarial} 
	We define an {\em adversarial} version of the classical weighted set cover problem:
Given a ground set $U$, a collection of subsets $\mathcal{S}\subseteq$ $2^U$, a weight budget $M$.
The objective of the {\em adversarial} weighted set cover problem is to allocate weight among subsets in $\cal{S}$ such that the minimum weight of all set covers is maximum.

	Formally, the objective of the {\em adversarial} weighted set cover problem is: \begin{center}
		$\max_{w: \sum_{S\in\mathcal{S}}w(S)\leq M}\min_{\mathcal{C}: \mathcal{C}~is~a~set~cover}\sum_{S\in\mathcal{C}}w(S)$, 
	\end{center}
where $w$ is a nonnegative allocation vector.

Similarly, we can define {\em adversarial} weighted vertex cover problem,  {\em adversarial} weighted dominant set problem, and so on.

	The following argument shows that the adversarial weighted set cover (dominant set, vertex cover, etc.) problem is a special instance of \rlcv.
	When $\eta=\sigma(V)=M$,
	\rleastcorevalue can be denoted more compactly:
	\begin{center}
		\rleastcorevalue$=\min_{x:x(V)=M, x\geq0}\max_{S:\sigma(S)=M}(1-\frac{x(S)}{M})$.
	\end{center}
	Thus, it is enough to compute
\begin{center}
		$\max_{x:x(V)=M, x\geq0}\min_{S:\sigma(S)=M} x(S)$.
\end{center}
	Given any instance of set cover problem,
	similar to the construction in the proof of Theorem \ref{thm:oraclehard},
	it is not difficult to construct a social graph such that
	$\inff(S)$ equals the number of elements covered by $S$, for any collection $S$.
	Thus, the adversarial weighted set cover problem is a special instance of \rlcv problem.
\end{remark}

\begin{proof}[Proof of Theorem \ref{thm:rleastcorenphard}]
	We construct a reduction from the SAT problem.
	A boolean formula is in conjunctive normal form (CNF) if it is expressed as an AND of clauses,
	each of which is the OR of one or more literals.
	The SAT problem is defined as follows:
	given a CNF formula $F$, determine whether $F$ has a satisfiable assignment.
	Let $F$ be a CNF formula with $m$ clauses $C_1, C_2, \cdots, C_m$, over $n$ literals $z_1, z_2, \cdots, z_n$.
	Without loss of generality, we set $m>4n$.
	
	We construct a social graph $G$ as follows:
	\begin{figure}[t]
		\centering
		\includegraphics[width=0.8\textwidth]{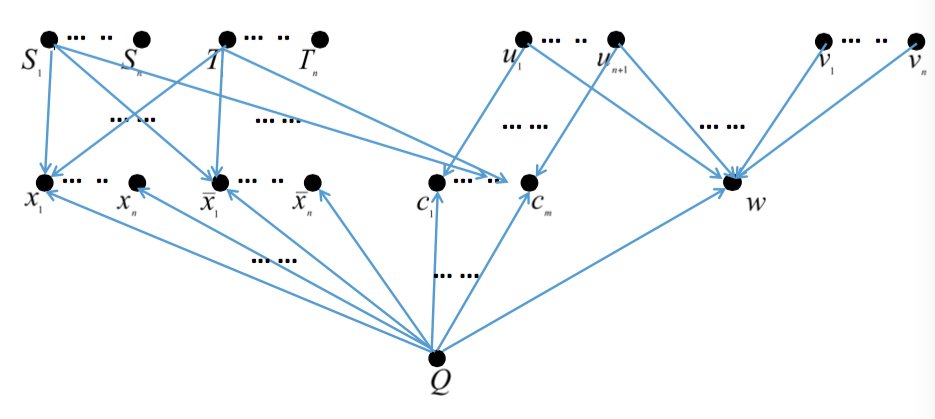}\\
		\caption{The reduction from SAT to \rleastcorevalue}\label{fig:reduction-sat_appendix}
	\end{figure}
	$G=(V_1\cup V_2\cup V_3, E)$ is a tripartite graph (see the sketch graph in Figure \ref{fig:reduction-sat_appendix}).
	In the first layer ($V_1$), there are two nodes $S_i$ and $T_i$ corresponding to each $i\in \{1, 2, \cdots, n\}$,
	$n+1$ dummy nodes labelled as $u_1, u_2, \cdots, u_{n+1}$ and $n$ dummy nodes labelled as $v_1, v_2, \cdots, v_n$.
	In the second layer ($V_2$), there are two nodes $x_i$ and $\overline{x}_i$ corresponding to each $i\in \{1, 2, \cdots, n\}$, one node $c_j$ for each $j\in\{1, 2, \cdots, m\}$ and a dummy node $w$.
	The third layer ($V_3$) contains only node $Q$.
	Edges exist only between the adjacent layers.
	For each $i\in \{1, 2, \cdots, n\}$, $S_i$ sends an edge to every node in
	$\{x_i, \overline{x}_i\}\cup\{c_j:$ clause $C_j$ contains literal $z_i, j\in\{1, 2, \cdots, m\}\}$.
	Similarly, for each $i\in \{1, 2, \cdots, n\}$, $T_i$ sends an edge to every node in
	$\{x_i, \overline{x}_i\}\cup\{c_j:$ clause $C_j$ contains literal $\overline{z}_i, j\in\{1, 2, \cdots, m\}\}$.
	The probabilities on edges sent form $S_i$ and $T_i$ are 1.
	There is an edge with influence probability 1 from $u_i$ to $c_i$ for any $i\in \{1, 2, \cdots, n\}$ and $m-n$ edges form $u_{n+1}$ to $c_{n+1}, c_{n+2}, \cdots, c_m$.
	There is an edge from $u_i$ to $w$ with influence probability $1-\sqrt[n+1]{1/2}$ for any $i\in \{1, 2, \cdots, n+1\}$.
	There is also exists an edge from $v_i$ to $w$ with influence probability $1-\sqrt[n]{1/2}$ for any $i\in \{1, 2, \cdots, n\}$.
	The left edges are  from $Q$ to all nodes in the second layer.
	The influence probability on edge $(Q, w)$ is $1/2$ and all other probabilities on edges sent from $Q$ is 1.
	The influence cooperative game defined on $G$ is $\Gamma(G)=(V=V_1\cup V_3, \inff(\cdot), \eta=2n+m+1/2)$.
	For convenient, we set $N=2n+m$.

	Under the above construction, if $F$ is satisfiable and the corresponding assignment is $\{y_1, y_2, \cdots, y_n\}$.
	Let $A=\{S_i: y_i=1, i\in \{1, 2, \cdots, n\}\}\cup \{T_i: y_i=0, i\in \{1, 2, \cdots, n\}\}$,
	$B=\{S_i: y_i=0, i\in \{1, 2, \cdots, n\}\}\cup \{T_i: y_i=1, i\in \{1, 2, \cdots, n\}\}$.
	Thus, $A$ can active all nodes in the second layer except $w$ and
	$B$ can active all nodes in $\{x_1, x_2, \cdots, x_n\}\cup \{\overline{x}_1, \overline{x}_2, \cdots, \overline{x}_n\}$.
	Therefore, there are three disjoint successful coalitions $A'=A\cup \{v_1, v_2, \cdots, v_n\}$, $B'=B\cup \{u_1, u_2, \cdots, u_{n+1}\}$ and $Q$ which means $\inff(A')\geq\eta$, $\inff(B')\geq\eta$ and $\inff(Q)\geq\eta$.
	Suppose $r^*$ is the optimal solution of the relative least-core value of  $\Gamma(G)$ and $x^*$ is an optimal allocation.
	We can prove $r^*\geq 1-\frac{1}{3}(N+\frac{7}{8})/(N+\frac{1}{2})$.
	This conclusion can be derived by separately considering cases $x^*(Q)> \frac{1}{3}(N+\frac{7}{8})$, $x^*(Q)<  \frac{1}{3}(N+\frac{7}{8})$ and $x^*(Q)= \frac{1}{3}(N+\frac{7}{8})$.
	
	When $F$ is un-satisfiable, it is sufficient for our proof if we can find an solution $(x, r)$
	such that $r<1-\frac{\frac{1}{3}(N+\frac{7}{8})}{N+\frac{1}{2}}$.
	Note that when $F$ is un-satisfiable, then for any $S\subseteq V$,  $|S|\geq 2n+1$ if $\sigma(S)\geq \eta$.
	Otherwise, we can construct an assignment such that $F$ is satisfiable.
	Let $x(Q)=\frac{1}{3}(N+\frac{7}{8})+\alpha$ and for any $v\in V$, $x(v)=\frac{\frac{2}{3}(N+\frac{7}{8})-\alpha}{4n+1}$.
	
	The left is to prove that there exists a positive $\alpha$ such that $\frac{x(Q)}{\inff(Q)}\leq \frac{x(S)}{\inff(S)}$ for any $S$ satisfying $\inff(S)\geq \eta$.
	We prove the above inequality by considering the following two cases:\\
	(i) $Q\in S$: In this case, 
	\begin{equation*}
	\begin{aligned}
	\frac{x(S)}{\inff(S)}&\geq \frac{x(Q)+x(S\setminus \{Q\})}{\inff(\{Q\})+ 1}\\
	&\geq \min\{\frac{x(Q)}{\inff(\{Q\})},x(S\setminus \{Q\})\}\\
	&\geq \min \{\frac{\frac{1}{3}(N+\frac{7}{8})+\alpha}{N+\frac{1}{2}}, \frac{\frac{2}{3}(N+\frac{7}{8})-\alpha}{4n+1}\}.
	\end{aligned}
	\end{equation*}
	There exists $\alpha>0$ such that $\frac{\frac{1}{3}(N+\frac{7}{8})+\alpha}{N+\frac{1}{2}}<\frac{\frac{2}{3}(N+\frac{7}{8})-\alpha}{4n+1}$
	since $\frac{\frac{1}{3}(N+\frac{7}{8})}{N+\frac{1}{2}}<\frac{\frac{2}{3}(N+\frac{7}{8})}{4n+1}$.
	Thus, 
\begin{center}
		$\frac{x(S)}{\inff(S)}\geq \frac{x(Q)}{\inff(\{Q\})}$.
\end{center}
	
	(ii) $Q\notin S$: In this case, there exists an $\alpha>0$ such that
	\begin{equation*}
	\begin{aligned}
	\frac{x(S)}{\inff(S)}&= \frac{|S|[\frac{2}{3}(N+\frac{7}{8})-\alpha]}{(4n+1)(N+\frac{3}{4})}\\
	&\geq \frac{(2n+1)[\frac{2}{3}(N+\frac{7}{8})-\alpha]}{(4n+1)(N+\frac{3}{4})}\\
	&>\frac{[\frac{1}{3}(N+\frac{7}{8})+\alpha/2]}{N+\frac{1}{2}}\\
	&=\frac{x(Q)}{\inff(Q)}.
	\end{aligned}
	\end{equation*}

\end{proof}

\section{Appendix of Section \ref{sec:leastcorevalue}}
\begin{proof}[Proof of Theorem \ref{thm:alcvhardness}]
	We construct a reduction from MAX-CUT problem .
	The input of MAX-CUT is an undirected graph $g=(W,E)$, where $W$ is the node set and $E$ is the edge set.
	The question is to compute the value of the maximum size of all edge cuts in $g$.
	Given an instance of MAX-CUT problem $A_1$, we construct an instance of influence cooperative game $A_2$ as follows:
	
	The social graph is a one-way layer graph $N=(V \cup U_1 \cup U_2, D)$,
	where each node in $V$ corresponds to a node in $W$ one-to-one and each node in $U_1$ corresponds to an edge in $E$ one-to-one.
	There is a directed edge from $v\in V$ to $u\in U_1$ if and only if
	the node in $W$ corresponding to $v$ is one of the endpoints of the edge in $E$ corresponding to $u$.
	$U_2$ is a copy of $U_1$ and there is a directed edge from each node in $U_1$ to its copy node in $U_2$.
	The probabilities on all edges equal 1.
	Figure \ref{fig:leastcorevalue_hardnessproof_appendix} shows an example of the above construction.
	
	The influence cooperative game defined on $N$ is $A_2=(V, \inff(\cdot), \eta=0)$.
	Thus, ${\cal ALCV}(A_2)$ equals the optimal solution of the following linear programming:
	
	\begin{equation}\label{equ:leastcorevalue_hardnessprove_appendix}
	\OnlyInShort{\small}
	\begin{array}{ll}
	\mbox{min}& \ \varepsilon\\
	\mbox{s.t.}&\left\{
	\begin{array}{ll}
	x(V)=\inff(V)\\
	x(S)\ge \inff(S)-\varepsilon  & \quad \forall~S\subseteq V\\
	x(\{u\}) \ge 0 & \quad \forall~u\in V\\
	\end{array}\right.
	\end{array}
	\end{equation}
	
	\begin{figure}[t]
		\subfigure[$g$: An under-graph of MAX-CUT problem]
		{
			\label{maxcut_appendix}
			\includegraphics[width=6cm]{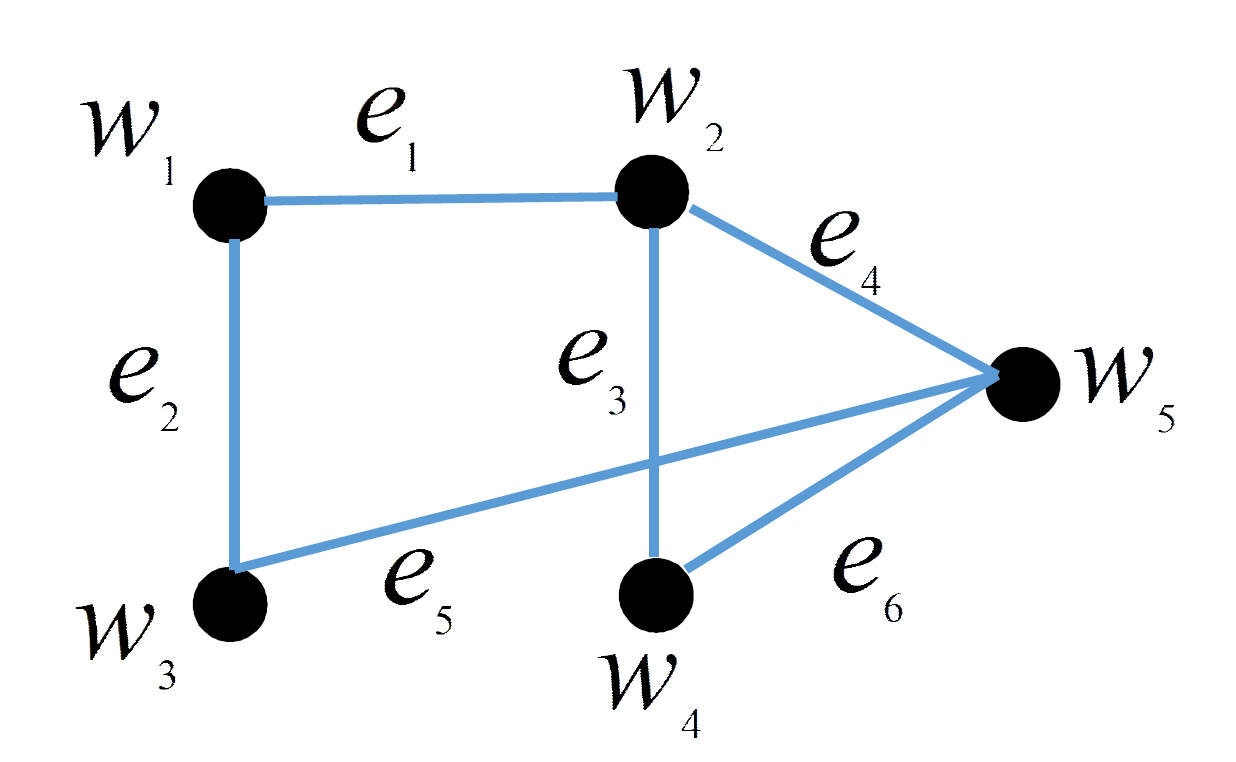}\\
		}
		\subfigure [$N$: The social graph transformed from $g$]
		{
			\includegraphics [width=6cm]{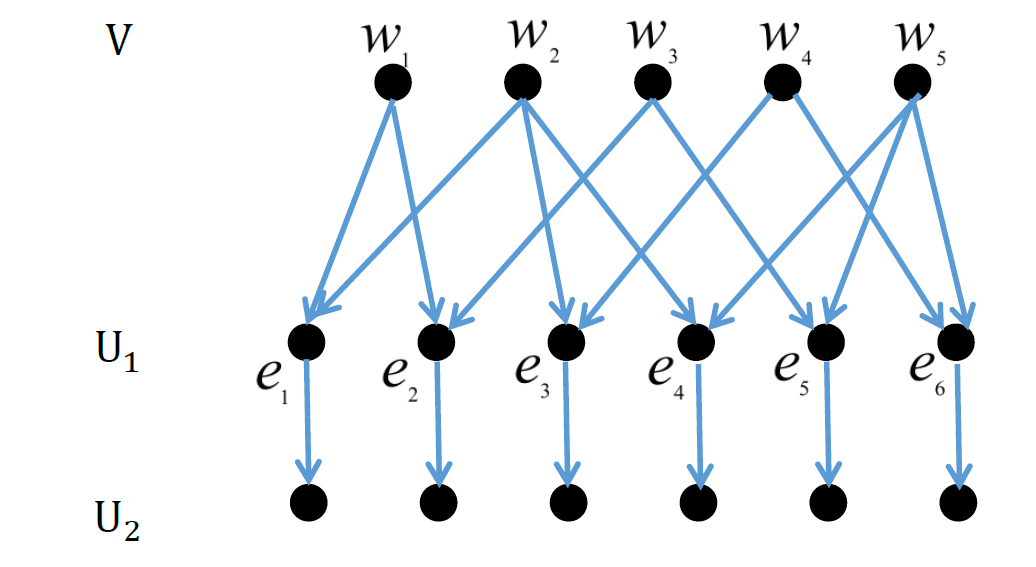} \\
			\label{fig:socialgraph_frommaxcut_appendix}
		}
		\caption {Constructing a social graph form a MAX-CUT instance}
		\label{fig:leastcorevalue_hardnessproof_appendix}
	\end{figure}
	
	In $g=(W,E)$, given any subset $S\subseteq W$, let ${\cal I}(S)=|\{(i,j)\in E: i\in S, j\in S\}|$
	and ${\cal C}(S)=|\{(i,j)\in E: i\in S, j\in W\setminus S\}|$.
	Clearly, ${\cal I}(S)$ is the number of edges induced by $S$ and ${\cal C}(S)$ is the size of the cut between $S$ and $W\setminus S$.
	To prove the hardness of computing ${\cal ALCV}(A_2)$, it is sufficient to prove  ${\cal ALCV}(A_2)=c^*$, where $c^*$ is the size of the maximum cut of $g$.
	Let $(x^*, \varepsilon^*)$ be the optimal solution of (\ref{equ:leastcorevalue_hardnessprove_appendix}).
	Therefore,
		\begin{equation*}
		\begin{cases}
		x^*(S)\geq \inff(S)-\varepsilon^*, \forall S\subseteq V\\
		x^*(V\setminus S)\geq \inff(V\setminus S)-\varepsilon^*, \forall S\subseteq V
		\end{cases}
		\end{equation*}
	Summing up these two inequalities, we have, $2\varepsilon^*\geq \inff(S)+\inff(V\setminus S)-x^*(V)
	=2{\cal I}(S)+2{\cal C}(S)+2{\cal I}(V\setminus S)+2{\cal C}(V\setminus S)-\inff(V)
	=2{\cal I}(V)+2{\cal C}(V\setminus S)-\inff(V)
	=2{\cal C}(V\setminus S)=2{\cal C}(S).$
	Then $\varepsilon^*\geq c^*$ since $\varepsilon^*\geq {\cal C}(S)$ for any $S\subseteq V$.
	
	Let $x=(x_1,x_2,\cdots,x_n)$ and $x_i$ be the degree of $i$ for any $i\in W$.
	It is obvious that $x(V)=2|E|=\inff(V)$.
	For any coalition $S\subseteq V$, $\inff(S)=2({\cal I}(S)+{\cal C}(S))
	=(2{\cal I}(S)+{\cal C}(S))+{\cal C}(S)
	=x(S)+{\cal C}(S)\leq x(S)+c^*.$
	That is to say $(x,c^*)$ is a feasible solution of (\ref{equ:leastcorevalue_hardnessprove_appendix}).
	Thus, $c^*\geq\varepsilon^*$.
	Combining with $c^*\leq\varepsilon^*$, we have $c^*=\varepsilon^*$.
	
	If there is a $\rho$-approximate ($\rho\geq 1$) algorithm ${\mathcal A}$ for computing ${\cal ALCV}(A_2)$ and ${\mathcal A}$ outputs $\varepsilon$.
	Therefore, $\varepsilon^*\leq\varepsilon\leq\rho\varepsilon^*$.
	Which means  $c^*\leq\varepsilon\leq\rho c^*$.
	Thus, $\varepsilon/\rho$ is a $(1/\rho)$-approximate value of $c^*$.
	If the unique games conjecture is true, \cite{khot2007optimal} proved that MAX-CUT cannot be approximated within
	$\alpha=\frac{2}{\pi}\min_{0\leq\theta\leq\pi}\frac{\theta}{1-\cos{\theta}}(\alpha\approx0.878)$.
	Thus, \alcv problem of influence cooperative game can not be approximated within $\frac{1}{\alpha}\approx 1.139$.
\end{proof}

\begin{proof}[Proof of Lemma \ref{lem:lpplprlps}]
	Suppose the optimal solutions of \lpprime, \lprelax and \lpstrengthen are $(x_p^*,\varepsilon_p^*)$, $(x_r^*,\varepsilon_r^*)$
	and $(x_s^*,\varepsilon_s^*)$, respectively.
	
	We first prove inequality (\ref{eq:lprandlpp}).
	It obvious that $(x_p^*,\varepsilon_p^*)$ is a feasible solution of \lprelax.
	Thus, $\varepsilon_r^*\leq \varepsilon_p^*$.
	In \lprelax, for any successful coalition $S$, we have
	$x_r^*(S)\geq\eta-\varepsilon_r^*\geq\sigma(S)-(\varepsilon_r^*+\sigma(V)-\eta)$.
	Therefore, $(x_r^*,\varepsilon_r^*+\sigma(V)-\eta)$ is a feasible solution of \lpprime.
	Thus, $\varepsilon_p^*\leq\varepsilon_r^*+\sigma(V)-\eta$ and then (\ref{eq:lprandlpp}) sets up.
	Now we prove (\ref{eq:lpsandlpp}). Similar to the proof of inequality (\ref{eq:lprandlpp}), $(x_s^*,\varepsilon_s^*)$ is a feasible solution of \lpprime.
	Thus, $\varepsilon_p^*\leq \varepsilon_s^*$.
	In \lpprime, for any successful coalition $S$, we have
	$x_p^*(S)\geq\sigma(S)-\varepsilon_p^*$.
	For any unsuccessful coalition $S'$, we have $\sigma(S')-x_p^*(S')<\eta$.
	Thus, $(x_p^*(S),\max(\varepsilon_p^*,\eta))$ is a feasible solution of \lpstrengthen.
	Thus, $\varepsilon_p^*\leq\varepsilon_s^*\leq\max(\varepsilon_p^*,\eta)$.
\end{proof}

\begin{proof}[Proof of Lemma \ref{lem:lpstrengthen}]
	\lpstrengthen captures cooperative games with submodular profit function.
	\cite{schulz2013approximating} proposed a framework to approximate the least-core value of this kind of cooperative games.
	
	They defined an optimization problem names as
	\emph{$x$-maximum dissatisfaction problem for cooperative game $(V, \sigma)$} ($x$-MD),
	where $V$ is the player set and $\sigma$ is the submodular profit function.
	The definition of $x$-MD is:
	Given any allocation $x$ such that $x(V)=\sigma(V)$,
	find a coalition $S^*$ whose dissatisfaction is maximum. i.e.
	$\max_{S\subseteq V}\{\sigma(S)-x(S)\}$.
	Under their framework, a $\rho$-approximation algorithm of $x$-MD implies a $1/\rho$-approximation algorithm
	of the least-core value of the cooperative game $(N,\sigma)$.
	Moreover, ones can find an allocation in this $\rho$-approximation least-core.
	
	Note that, given an allocation $x(V)=\sigma(V)$, finding the maximum value of $\sigma(S)-x(S)$ falls into
	the \emph{submodular function maximization problem} since $\sigma(S)-x(S)$ is submodular.
	In \cite{buchbinder2015tight}, the authors design a deterministic 1/3-approximate algorithm of submodular function maximization problem when the function value on $\emptyset$ and $V$ is nonnegative.
	That is to say, there exists an 1/3-approximation algorithm of \lpstrengthen since $\sigma(\emptyset)-x(\emptyset)=\sigma(V)-x(V)=0$.
\end{proof}

\begin{proof}[Proof of Lemma \ref{lem:lprelax}]
	The outline of the dynamic scheme is a generalization of the process in \cite{elkind2007computational}
	In \lprelax, let $\alpha=\eta-\varepsilon$, then \lprelax can be transformed to the following linear programming (\lprelaxnew):
	\begin{equation}\label{equ:lprelax'}
	\begin{array}{ll}
	\mbox{max}& \ \alpha\\
	\mbox{s.t.}&\left\{
	\begin{array}{ll}
	x(V)=\sigma(V)\\
	x(\{S\})\ge \alpha  & \quad \forall S\subseteq V, \sigma(S)\geq \eta\\
	x(\{u\}) \ge 0 & \quad \forall u\in V\\
	0\leq \alpha \leq \eta
	\end{array}\right.
	\end{array}
	\end{equation}
	We break \lprelaxnew into a family of linear feasibility programs ${\cal F}=\{LFP_1, LFP_2,\cdots, LFP_t\}$,
	where $t=\lceil\frac{\eta}{\delta}\rceil$.
	The $k$-th linear feasibility program $LFP_k$ is:
	\begin{equation*}
	\begin{cases}
	x(V)=\sigma(V)\\
	x(\{S\})\ge k\delta  & \quad \forall S\subseteq V, \sigma(S)\geq \eta\\
	x(\{u\}) \ge 0 & \quad \forall u\in V\\
	\end{cases}
	\end{equation*}
	Let $k^*=\max\{k:LFP_k$ has a feasible solution$\}$ and $\alpha^*$ be the optimal value of \lprelaxnew.
	Thus, $k^*\delta\leq \alpha^*<(k^*+1)\delta$.
	
	To solve these linear feasibility programs, we need a polynomial separation oracle since there are exponential constrains in each program.
	However, we do not know how to construct such an oracle.
	Following the idea in \cite{elkind2007computational}, we also construct a ``partial'' separation oracle ${\cal O}_k$ for each program $LFP_k$, $k=1,\cdots,t$.
	${\cal O}_k$ has the following property, given a candidate solution $x$ to $LFP_k$,
	the output of ${\cal O}_k$ falls into one of the following cases:
	\begin{enumerate}
		\item ${\cal O}_k$ successes, i.e. it can assert that $x$ is a feasible solution for $LFP_k$ or output a violate constraint in $LFP_k$.
		\item ${\cal O}_k$ fails, then it outputs a feasible solution for $LFP_{k-1}$.
	\end{enumerate}
	
	Now we introduce how to use this partial separation oracle.
	Suppose we run the ellipsoid algorithm for each $LFP_k$, $k=1,\cdots,t$, using ${\cal O}_k$ instead of a proper separation oracle.
	If ${\cal O}_k$ successes, we can obtain a feasible solution to $LFP_k$ or assert that $k^*=k-1$.
	If ${\cal O}_k$ fails, we can obtain a feasible solution to $LFP_{k-1}$.
	Thus, when we work to $k=k^*$, we can obtain a feasible solution of $LFP_k^*$ or $LFP_{k^*-1}$.
	Let $k'$ be the largest value of $k$ for which our procedure finds a feasible solution for $LFP_k$.
	It holds that $k^*-1\leq k'\leq k^*$.
	Now, it is not difficult to derive that $\varepsilon^*_r\leq \eta-k'\delta\leq\varepsilon^*_r+2\delta$.
	
	A crucial problem is how to design ${\cal O}_k$.
	Indeed the main idea of ${\cal O}_k$ is to apply dynamic programming such that
	we can decide whether there exists a coalition $S$ with $\sigma(S)\geq \eta$ but $x(S)<k\delta$ under the given $x$.
	That is to say, given $x(V)=\sigma(V)$, we want to compute $\max\{\sigma(S)|x(S)<k\delta\}$.
	To use dynamic programming, we need to discretized $x(S)$.
	For each $i\in V$, let $x'_i=\max\{j\delta'\leq x_i\}$, where $\delta'=\frac{\delta}{M}$.
	Now, for each $j\in \{1,\cdots, n\}$ and $l\in {\cal L}\{0,\cdots,(k-1)M-1\}$,
	let $z[j,l]=\max\{\sigma(S)|S\subseteq\{1,\cdots,j\}, x'(S)=l\delta'\}$.
	For each $l\in {\cal L}$, we initialize $z[j,l]$ as:
	\begin{equation*}
	z[1,l]=
	\begin{cases}
	\sigma(\{1\})&\text{$l=\lfloor\frac{x'_1}{\delta'}\rfloor$}\\
	-\infty & \text{otherwise}\\
	\end{cases}
	\end{equation*}
	The iteration rule is:
	\begin{center}
		$z[j,l]=\max\{z[j-1,l], z[j-1,\frac{l\delta'-x'_j}{\delta'}+\sigma(\{1,\cdots,j\})-\sigma(\{1,\cdots,j-1\})]\}$.
	\end{center}
	Let $U=\max\{z[n,i],l\in{\cal L}\}$, then we can find a violated constrains if $U\geq \eta$.
	Otherwise, if $U\leq \eta$, then for any $S\in V$ with $\sigma(S)\geq \eta$, we have $x'(S)\geq(k-1)n\delta'-\delta'$,
	thus, $x(S)\geq(k-1)\delta$.
	That is to say, $x$ is a feasible solution to $LFP_{k-1}$.
	The running time of ${\cal O}_k$ is polynomial time in $\log \sigma(V)$, $n$ and $1/\delta$.
\end{proof}

\section{Appendix of Section \ref{sec:ladv}}
\begin{proof}[Proof of Lemma \ref{lem:gsubmodular}]
	\begin{equation*}
	\begin{aligned}
	&tF(x^{(1)})+(1-t)F(x^{(2)})\\
	=&\frac{1}{2^n}\sum_{S\subseteq V}(\max (t(f(S)-x^{(1)}((S)), 0)+\\
	&\max ((1-t)(f(S)-x^{(2)}((S)), 0))\\
	\geq &\frac{1}{2^n}\sum_{S\subseteq V}\max (t(f(S)-x^{(1)}(S))+(1-t)(f(S)-x^{(2)}(S)), 0)\\
	=&F(tx^{(1)}+(1-t)x^{(2)}).	
	\end{aligned}
	\end{equation*}
\end{proof}

\newpage
\bibliographystyle{abbrv}
\bibliography{inf-game}

\end{document}